\pdfoutput=1
\documentclass[11pt]{amsart}
\usepackage{amsmath}
\usepackage{amssymb}
\usepackage{amsthm}
\usepackage{bbm}
\usepackage{color}
\usepackage[colorlinks,linkcolor=black,citecolor=black,linktocpage,hypertexnames=true,pdfpagelabels]{hyperref}
\usepackage[capitalise]{cleveref}
\usepackage[all]{xy}
\usepackage{subcaption}
\usepackage{array}
\usepackage{enumerate}
\usepackage{enumitem}
\usepackage{tikz}
\usepackage{relsize}
\usepackage{diagbox}
\usepackage{tikz-cd} 
 \usepackage[foot]{amsaddr}
 
 \usepackage[full]{textcomp} 
 
\newtheorem{theorem}{Theorem}

\newtheorem{proposition}[theorem]{Proposition}

\newtheorem{definition}[theorem]{Definition}
\theoremstyle{definition}

\newtheorem{observation}[theorem]{Observation}
\newtheorem{example}[theorem]{Example}
\newtheorem{remark}[theorem]{Remark}

\renewcommand{\epsilon}{\varepsilon}
\newcommand{\N}{\mathbb{N}}
\newcommand{\C}{\mathbb{C}}

\newcommand{\be}{\begin{eqnarray}}
\newcommand{\ee}{\end{eqnarray}}
\newcommand{\ben}{\begin{enumerate}}
\newcommand{\een}{\end{enumerate}}

\newcommand{\ba}{\begin{array}}
\newcommand{\ea}{\end{array}}

\newcommand{\Her}{\mathrm{Her}}

\newcommand\setn{\{1,\ldots, n\}}

\newcommand{\hern}{\mathrm{Her}_n(\mathbb{C}) }
\newcommand{\hers}{\mathrm{Her}_s (\mathbb{C})}

\newcommand{\mat}{\mathrm{Mat}}

\renewcommand{\phi}{\varphi}
\renewcommand{\epsilon}{\varepsilon}
\renewcommand{\theta}{\vartheta}

\newcommand\mconv{{\rm mconv}}

\usepackage{xfrac}
\usepackage{mathtools}

\usepackage[textsize=footnotesize]{todonotes}

\let\originalleft\left
\let\originalright\right
\renewcommand{\left}{\mathopen{}\mathclose\bgroup\originalleft}
\renewcommand{\right}{\aftergroup\egroup\originalright}

\usepackage{kpfonts}

\usepackage{pgfornament}

\begin{document}
\title{Magic squares: Latin, Semiclassical and Quantum} 

\begin{abstract}
Quantum magic squares were recently introduced as a `magical' combination of quantum measurements. In contrast to quantum measurements, they cannot be purified (i.e.\ dilated to a quantum permutation matrix)---only the so-called semiclassical ones can. Purifying establishes a relation to an ideal world of fundamental theoretical and practical importance; the opposite of purifying is described by the matrix convex hull. In this work, we prove that semiclassical magic squares can be purified to quantum Latin squares, which are `magical'  combinations of orthonormal bases. Conversely, we prove that the matrix convex hull of quantum Latin squares is larger than the semiclassical ones. This tension is resolved by our third result: We prove that the quantum Latin squares that are semiclassical are precisely those constructed from a classical Latin square. Our work sheds light on the internal structure of quantum magic squares, on how this is affected by the matrix convex hull, and, more generally, on the nature of the `magical' composition rule, both at the semiclassical and quantum level. 
\end{abstract}

\author{Gemma De las Cuevas}
\address{Institute for Theoretical Physics, Technikerstr.\ 21a,  A-6020 Innsbruck, Austria}

\author{Tim Netzer}
\address{Department of Mathematics, Technikerstr.\ 13,  A-6020 Innsbruck, Austria}

\author{Inga Valentiner-Branth}
\address{Department of Mathematics, Technikerstr.\ 13,  A-6020 Innsbruck, Austria}

\date{\today}

\maketitle

\section{Introduction }

Magic squares  have fascinated mathematicians and non-mathematicians for more than 2000 years. They are defined as square matrices filled with nonnegative numbers, so that the entries in each row and column (and sometimes also the diagonal) sum to the same number (the so-called {\it magic constant} of the square).  If the magic constant is one (and there is no condition on the diagonal), they are sometimes called {\it doubly stochastic matrices}. We will use the notions of magic square and doubly stochastic matrix interchangeably in this work.

A quantum version of magic squares has been introduced in \cite{J} (see also \cite{De21d} for an invitation to this work),  
where the entries are no longer nonnegative numbers, but positive semidefinite matrices, and every row and column of the {\it quantum magic square} sums to the identity matrix. 
That is, rows and columns form a {\it positive operator valued measure} (POVM), which is a description of a quantum measurement. 
In other words: if we consider the POVMs defined by the rows, they can be combined into a quantum magic square only if the first elements of every POVM also define a POVM (i.e.\ the first column), and similarly with the second elements, and so on, until the last elements. 
This is the `magical' combination which, to the best of our knowledge, does not correspond to any physical composition (although see the considerations in the Outlook).  
Some features of this composition rule are 
that only POVMs with the same number of outcomes can be composed into a quantum magic square, 
and that the order of outcomes in every POVM matters. 
Note that the definition of a quantum magic square is not a construction but a \emph{condition}---having  positive semidefinite elements that sum to the identity along every row and column. 
This work, together with others, attempts to characterise some aspects of the set of matrices satisfying this condition. 

Now, about 300 years ago {\it Latin squares} were introduced, among others to help construct magic squares. 
A Latin square is an $n\times n$-matrix filled with the numbers $1,\ldots, n$ such that each number appears exactly once in each row and column. 
A special case of Latin squares of size $9$ are known as Sudokus, for which there are additional constraints. 
Latin squares have several applications in mathematics, for example in the design of experiments \cite{bailey_2008} and as multiplication tables of quasigroups \cite{Sm06}. 
Clearly, every Latin square is a non-normalized magic square.

There are several quantum generalisations of Latin squares. 
On the one hand, a quantum version of Sudokus has recently been proposed in \cite{Ne20}. 
On the other hand, a quantum version of Latin squares was introduced in \cite{Mu16b}, where instead of numbers as entries in the Latin square,  (pure) quantum states are used. 
In a quantum Latin square, every entry is a unit vector from $\mathbb{C}^n$ such that every row and column forms an orthonormal basis. 
 In \cite{Mu16b} it is shown that the so-called quantum shift-and-multiply method can be used to construct unitary error bases from quantum Latin squares. Unitary error bases are special bases of the set of complex matrices of size $n\times n$, which are widely used in quantum information theory \cite{UEB}, the most famous example being the Pauli matrices. 

In this work, we study quantum Latin squares and quantum magic squares from a unified perspective, 
and examine their relations.
We prove the following: 
\begin{enumerate}
\item 
Every semiclassical magic square can be purified to a quantum Latin square. 
Equivalently, the matrix convex hull of several sets of semiclassical Latin squares is
the set of  semiclassical magic squares (\Cref{thm_mconv}). 
\item 
The matrix convex hull of all quantum Latin squares is  larger than the set of semiclassical magic squares. This follows from the existence of non-semiclassical quantum Latin squares of even sizes larger than $2$ (\Cref{thm_ntqls}). 
\item  
The quantum Latin squares that arise directly from a classical Latin square are precisely the semiclassical ones (\Cref{prop:sc and ro = Ln}). 
\end{enumerate}

From a mathematical perspective, our work provides new insights into quantum magic squares by leveraging free convexity, and is part of an ongoing effort to establish synergies between free semialgebraic geometry and quantum information, see e.g.\ \cite{De21d,Bl18,Ji20,Bl21,Bu21}.

From a conceptual perspective, 
we characterise quantum magic squares by studying their relation to the ideal world, as the latter admits---presumably---a simpler description. (See \cite{Cu21} for related conceptual considerations). 
While in theoretical physics the ideal world typically involves an infinity (the thermodynamic limit, zero temperature, ground states, etc), 
this is not the case with regard to the description of quantum states, 
precisely because of the purification theorem, 
which expresses the relation to the ideal world as a relation between the parts and the whole, 
where this whole is finite (if the parts are finite).  
By virtue of this theorem, 
quantum states, 
quantum measurements 
and quantum channels can be purified,
that is to say, 
related to their respective ideal quantities 
(pure states, projective measurements and unitary maps) 
of finite dimension (assuming we start from a finite dimension).  
Mathematically, these purifications follow from Stinespring's dilation theorem, 
which says that every completely positive map is a multiplicative map followed by a contraction; 
the multiplicative maps are easy to characterise and here play the role of the ideal map, 
whereas the contraction plays the opposite role to the purification.  
In short, 
purifying corresponds to relating to the ideal world, 
whereas taking the matrix convex hull 
corresponds to a compression to our (non-idealised, imperfect) world. 

The purification theorem applies to every row and column of the quantum magic square, 
i.e.\ it establishes relations to the ideal world for every such `one dimensional' array of objects. 
This work investigates the nature of the purification in `two dimensions', 
namely we study which quantum magic squares can be purified, i.e.\ admit a purification that applies to all rows and columns \emph{simultaneously}. 

Other works related to this study include the quantum generalisations of magic squares  proposed in \cite{Me09b},
and the quantum permutation matrices  considered in \cite{Lu17b}. 
The recent resolution of Euler's officers problem in the quantum case involves the construction of a special  quantum magic square \cite{Ra21}. 

This paper is structured as follows.
In \Cref{sec_prel} we present quantum Latin and magic squares,  
in \Cref{sec:QLS vs QMS} we prove our main results,  
and in \cref{sec:conclusions} we conclude and provide an outlook. 

\section{Setting the stage}\label{sec_prel}

Throughout this work, $\mat_n(S)$ denotes the set of $n\times n$-matrices with entries from the set $S$. Given a matrix $A \in \mat_n(\C)$, $A^*$ denotes the complex conjugate of $A$, and we denote by $\Her_n(\mathbb C)=\{A\in\mat_n(\C)\mid A^* = A\}$ the real vector space of complex Hermitian $n\times n$ matrices. 
$A\geqslant 0$ denotes that the Hermitian matrix $A$ is positive semidefinite, and the convex cone of all positive semidefinite matrices is denoted by ${\rm Psd}_n(\C)$. 
The identity matrix of size $s$ is denoted by $I_s$. 

We will first revisit quantum magic squares and friends (\cref{ssec:qms}),
then define and explain the matrix convex hull (\cref{ssec:mconv}),
and finally consider quantum Latin squares (\cref{ssec:latin}). 

\subsection{Quantum magic squares and  friends} \label{ssec:qms}

Let us start by reviewing some basic definitions and results on quantum magic squares from \cite{J}.

First recall that 
a \emph{positive operator valued measure} (POVM) is a set of positive semidefinite matrices $A_1, \ldots, A_n \in {\rm Psd}_s(\mathbb C)$ such that $\sum_{i=1}^n A_i = I_s$. 
If every $A_i$ is a  projection, i.e.\ $A_i^2=A_i=A_i^*$, then the POVM is called a \emph{projective valued measure} (PVM). 

A quantum magic square is an $n\times n$ grid where every cell contains an $s\times s$ positive semidefinite matrix such that every row and column sums to the identity. 
For this reason we refer to $n$ as the \emph{external size} and to $s$ as the \emph{internal size}. 
A quantum permutation matrix is a quantum magic square where every element $A_{ij}$  is a projection, and in a commuting quantum permutation matrix we additionally require that $A_{ij}A_{k\ell} = A_{k\ell} A_{ij}$ for all $i,j,k,\ell\in \{1,\ldots,n\}$.

 \begin{definition}[Quantum magic squares and friends]
 Given exterior size $n$ and interior size $s$, 
 \begin{enumerate}
\item A \emph{quantum magic square} is a matrix $A \in \mat_n(\Her_s(\C))$ such that every row and  column of $A$ forms a POVM; 

\item 
A \emph{quantum permutation matrix} is a quantum magic square where every row and column forms a PVM;  

\item
A  \emph{commuting quantum permutation matrix} is a quantum permutation matrix where all entries commute. 
\end{enumerate}
\end{definition}

Fixing the interior size to 1 results in the classical matrices we are acquainted with. 
Specifically, quantum magic squares of interior size 1 are precisely the doubly stochastic matrices. 
Moreover, quantum permutation matrices of interior size 1 are the permutation matrices, 
since the only projectors in $\C$ are 0 and 1, 
and the magic square condition ensures that there is exactly one 1 entry per row and column.

Quantum magic squares with exterior size 1 and 2 are also very simple: 

\begin{remark}[Exterior size 1 and 2] \label{rem:extersize12} 
For exterior size $n=1$, the only possible quantum magic square is 
$A = (I_s)$.
For $n=2$ all quantum magic squares have the form
$$\begin{pmatrix}
A & I_s-A \\
I_s - A& A
\end{pmatrix} = \begin{pmatrix}
1 &0\\ 0&1
\end{pmatrix} \otimes A + \begin{pmatrix}
 0&1 \\ 1 &0
\end{pmatrix} \otimes (I_s-A),$$
where $A\in {\rm Psd}_s(\C)$ is such that also $I_s-A$ is positive semidefinite.
\end{remark}

Note that a quantum permutation matrix is a quantum representation \cite{Fr16b} of the hypergraph obtained by considering a $n\times n$ grid and defining a hyperedge for every row and column, as shown in \cref{fig:hypergraph}. 
The interior size of the quantum permutation matrix is precisely the dimension of the Hilbert space in \cite[Definition 7]{Fr16b}. 

\begin{figure}[t]
\includegraphics[width=.25\columnwidth]{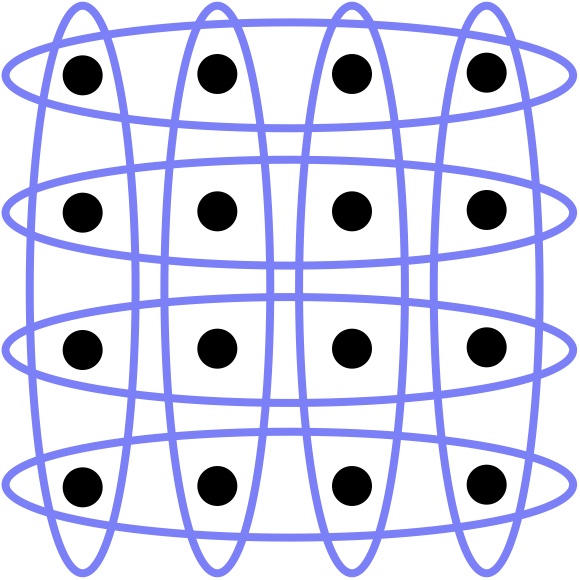}
\caption{A quantum permutation matrix is a quantum representation of the hypergraph shown here for $n=4$.}
\label{fig:hypergraph}
\end{figure}

We now turn  to \emph{semiclassical} magic squares, which were introduced in \cite{J}.
In \Cref{sec:QLS vs QMS} we will shed light on this notion from a broader perspective. 
We denote the permutation group on $n$ elements by $S_n$.  

\begin{definition}[Semiclassical]\label{def:semi} 
$A   \in \mat_n(\Her_s(\C))$ is a 
 \emph{semiclassical magic square} if 
  there exists a POVM $\{Q_{\pi}\}_{\pi \in S_n}$ such that 
\begin{equation}\label{eq:semiclassical}
A = \sum_{\pi \in S_n} P_{\pi} \otimes Q_{\pi},
\end{equation}
where $P_\pi\in{\rm Mat}_n(\C)$ is the permutation matrix corresponding to permutation $\pi$.
\end{definition}

A semiclassical magic square is a quantum magic square, as one can easily verify.

Semiclassical magic squares are clearly not classical, because they contain positive semidefinite matrices as its entries, but other than that they are  essentially classical, 
as they contain positive semidefinite matrices 
$$Q_1,\ldots, Q_{n!}$$ 
permuted through all inner positions. 
Specifically, in a semiclassical magic square, the elements of a single POVM $\{Q_\pi\}$ are summed and arranged according to the permutation matrices $P_{\pi}$, so as to form a quantum magic square.  
In this sense semiclassical magic square provide a constructive `magical' combination of POVMs stemming from a single, `primary' POVM.

Now, every magic square (a.k.a.\ doubly stochastic matrix) is semiclassical---this is precisely the content of Birkhoff--von Neumann's  Theorem.
Moreover, every quantum magic square of exterior size $n=1,2$ is semiclassical---this follows from  \cref{rem:extersize12}.

Let us now consider the union over all internal sizes of 
the set of quantum magic squares and their friends. 
This union is natural in the context of free semialgebraic geometry. 

\begin{definition}[Sets of quantum magic squares and friends]\label{def:cupM}
We denote the set of 
\begin{enumerate}
\item Quantum magic squares by
$$
 \mathsf{M}_{s}^{(n)}  \coloneqq \left\{ A\in {\rm Mat}_n\left({\rm Psd}_s(\C)\right) \mid 
\sum_{i} A_{ij} = \sum_{j} A_{ij} = I_s  \right\} ; 
$$ 
\item Quantum permutation matrices by
$$ \mathsf{P}_{s}^{(n)}  \coloneqq \left\{ A\in \mathsf{M}_{s}^{(n)} \mid 
A_{ij}^2 = A_{ij} =A_{ij}^* \mbox{ for every } i,j \right\} ; 
$$
\item Commuting quantum permutation matrices by
$$
 \mathsf{C}_{s}^{(n)} \coloneqq \left\{ A\in \mathsf{P}_s^{(n)} \mid 
A_{ij}A_{k\ell} = A_{k\ell}A_{ij} \, \mbox{ for every } i,j,k,\ell \right\} ; 
 $$
\item  
Semiclassical magic squares  by 
$$\mathsf{S}_{s}^{(n)} \coloneqq \left\{ A\in \mathsf{M}_{s}^{(n)}  \mid 
A \mbox{ is semiclassical }\right\}; 
$$
\end{enumerate}
and the union over all internal sizes, respectively, by 
$$
\mathsf{M}^{(n)}\coloneqq \bigcup_{s \in \N} \mathsf{M}_{s}^{(n)}, \ 
\mathsf{P}^{(n)}\coloneqq \bigcup_{s \in \N} \mathsf{P}_{s}^{(n)}, \ 
\mathsf{C}^{(n)}\coloneqq \bigcup_{s \in \N} \mathsf{C}_{s}^{(n)}, \ 
\mathsf{S}^{(n)}\coloneqq  \bigcup_{s \in \N} \mathsf{S}_{s}^{(n)}.
$$
\end{definition}

Note that, by definition, 
$$
\mathsf{C}^{(n)}  \subseteq \mathsf{P}^{(n)} \subseteq \mathsf{M}^{(n)} .
$$
In addition, for exterior size $n=1,2,3$, the commuting requirement makes no difference, i.e.\ 
$$
\mathsf{C}^{(n)} = \mathsf{P}^{(n)} .
$$
For $n=1,2$, this follows from \Cref{rem:extersize12}, and the case $n=3$ can be found in \cite{Lu17b}. 
On the other hand, for $n \geqslant 4$, these two sets are different, 
$$
\mathsf{C}^{(n)} \subsetneq \mathsf{P}^{(n)},
$$ 
as can be seen by taking block diagonal sums of quantum permutation matrices.
In addition, in \cite{J} it is shown that commuting quantum permutation matrices are semiclassical, 
$$ 
\textsf{C}^{(n)}\subseteq \textsf{S}^{(n)}.
$$

\subsection{The matrix convex hull} \label{ssec:mconv}

\begin{figure}[t]
\includegraphics[width=.9\columnwidth]{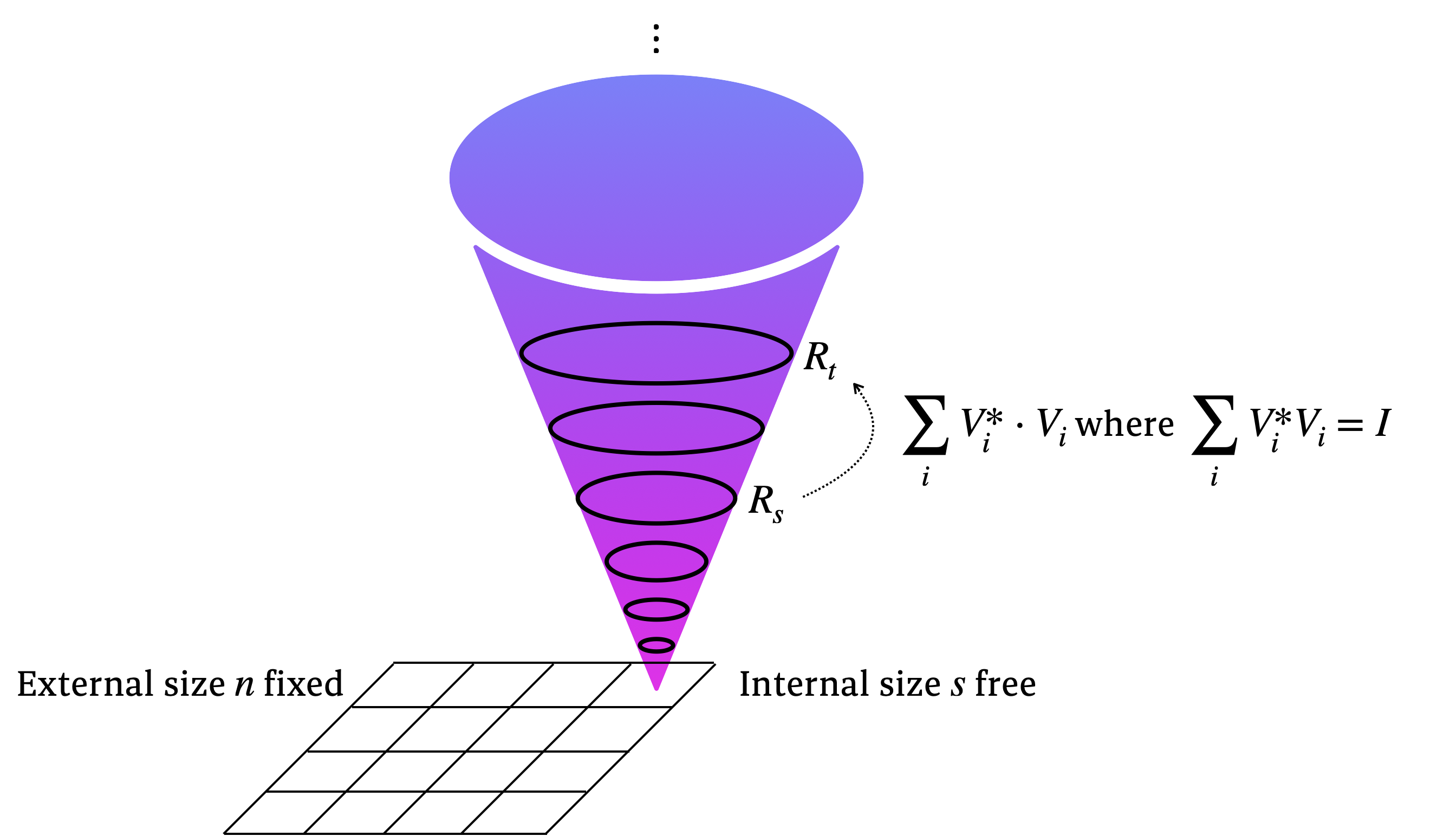}
\caption{The idea of matrix convexity. The external size $n$ remains fixed (in this depiction, $n=4$), whereas matrix convexity expresses a relation between all internal sizes $s$. Specifically, the sets $R_s$ at different internal level must be closed under contractions, which are transformations of the internal matrices with $\sum_{i} V_i^* \cdot V_i$ where $\sum_{i}V_i^*V_i = I$.}
\label{fig:matrixconvex}
\end{figure}

We now consider a dimension-free notion of convexity, called \emph{matrix convexity}, often used in  free semialgebraic geometry and operator algebra, see for example  \cite{Mconvfreesets, paulsen2002completely}.
The matrix convex hull plays a central role in this work, and the idea is the following (see \cref{fig:matrixconvex}). 
We are given a subset of matrices of external size $n$ and internal size $s$,  
$$
R_s \subseteq \mat_n(\hers) 
$$ 
for every $s$, 
and we consider 
$$
R = \bigcup_{s\in \N} R_s.
$$ 
We want to characterise what it means for $R$ to be matrix convex. 
Intuitively, it means that it is closed under a certain notion of contractions, 
which we will now explain step by step. 

Within a given level $s$, 
we consider a set of matrices $A^{(1)},\ldots, A^{(r)} \in R_s$ of internal size $s$ and external size $n$. 
We then consider the $k,\ell$ entry of each of these matrices, and contract them as 
$$
A_{k\ell}^{(1)} ,\ldots, A_{k\ell}^{(r)}  \mapsto \sum_{i=1}^r V_i^* A_{k\ell}^{(i)} V_i  
$$ 
where the $V_i$s are square matrices of size $s$ (i.e.\ the same size as $A_{k\ell}^{(i)}$, which is the internal size of $A$), which satisfy 
$$
\sum_{i=1}^r V_i^* V_i =I_s .
$$
The latter condition guarantees that the $V_i$s  behave like probabilities   in the matrix case, 
since if $s=1$, we recover  the usual notion of convexity.
Matrix convexity requires that, for a given level $s$, 
for any set $A^{(1)},\ldots, A^{(r)} \in R_s$ (with any $r$), 
the result of contracting with any such $V_i$s is still in $R_s$, 
$$
\sum_{i=1}^r V_i^*A^{(i)}V_i\coloneqq \left( \sum_{i=1}^r V_i^*A_{k,\ell}^{(i)}V_i \right)_{k,\ell=1}^n \in R_s. 
$$
Note that the contraction is applied to every cell (labeled $k,\ell$) separately.

Additionally, we can communicate different internal levels $s$. 
This is achieved by letting $V_i$ be rectangular matrices. 
Namely, if $R$ is matrix convex, 
for any levels $s$ and $t$, 
and for any matrices 
$$
V_i \in \mat_{s,t}(\C)  \mbox{ for }i=1,\ldots r \mbox{ with } \sum_i V_i^*V_i =I_t
$$
it must hold that 
$$
\sum_{i=1}^r V_i^*A^{(i)}V_i \in R_t. 
$$ 

In words, we can contract  the matrices of a given level $s$ to a smaller or a larger level $t$, and not leave $R$. 

Finally, in the most general case,  the initial matrices can be taken from different levels, 
$A^{(i)} \in R_{s_i}$ (with $i$ running from 1 to an arbitrary $r$), 
and we take the corresponding matrices $V_i$ to match this initial size $s_i$ and result in size $t$, 
$V_i \in \mat_{s_i,t}(\C)$. 

Note that the external size of the matrices $n$ is fixed through these operations. 
It is only the internal size of the matrices $s$ that is asked to satisfy these closure properties; 
specifically, every cell of the external $n\times n$ matrix is asked to satisfy these conditions.

\begin{definition}[Matrix convex] \label{def:mconv} $\quad$
\begin{enumerate} 
\item Let 
$$
R = \bigcup_{s\in \N} R_s \mbox{ where } R_s \subseteq \mat_n(\hers) . 
$$
$R $ is \emph{matrix convex} if 
for all $r,s_i,t \geqslant 1$,  
for all $A^{(i)}\in R_{s_i}$, 
and 
for all $V_i \in \mat_{s_i,t}(\C)$ with $i=1,\ldots, r$ 
and  $\sum_i V_i^*V_i = I_{t}$,
it holds that 
$$
\sum_{i=1}^r V_i^*A^{(i)}V_i \in R_t.
$$

\item 
For any set $R$, its \emph{matrix convex hull}, denoted ${\rm mconv}(R)$, is the smallest matrix convex superset of $R$, i.e.\ the intersection of all its matrix convex supersets. 
\end{enumerate}
\end{definition}

The matrix convex hull can thus be understood as a `free' version of the convex hull. 
Here free can either mean free of the dimension, or free of the commutation relation---both apply. 

One of the main results from \cite{J} states that the matrix convex hull of the quantum permutation matrices is smaller than the set quantum magic squares, 
\begin{equation} \label{frompaperj}
\mconv(\mathsf{P}^{(n)}) \subsetneq  \mathsf{M}^{(n)} \qquad \mbox{ for all }n\geqslant 3. 
\end{equation}
This is seen as a failure of a matrix version of the Birkoff--von Neumann theorem, since for internal size $s=1$  the convex hull of the set of permutation matrices equals the set of doubly stochastic matrices. 
\eqref{frompaperj} thus says that the characterisation of larger interior sizes is more difficult; 
at least, quantum magic squares cannot be characterised as the matrix convex hull of quantum permutation matrices. 

On the other hand, the matrix convex hull of the \emph{classical} (i.e.\ usual) permutation matrices is precisely the set of semiclassical  magic squares \cite{J},  
\begin{equation}\label{eq:mconvPn1}
\mconv(\textsf{P}^{(n)}_1) = \textsf{S}^{(n)}. 
\end{equation}
In other words, the semiclassical magic squares (elements $A$ in \eqref{eq:semiclassical}) 
 are precisely what we need to add to the usual permutation matrices so that the set is matrix convex. 
In this sense, semiclassical magic squares are a free convex promotion of classical permutation matrices---this is yet another way in which they are semiclassical. 

Overall, we have  
$$
\mconv(\textsf{C}^{(n)}) =  {\mconv}(\textsf{P}^{(n)}_1) = \textsf{S}^{(n)}  
$$ 
because $ \textsf{C}^{(n)}$ is contained in  $\textsf{S}^{(n)}$  and contains $\textsf{P}^{(n)}_1$.

\subsection{Quantum Latin squares} \label{ssec:latin}

We now turn our attention to (quantum)  Latin squares. 
A Latin square  is a square where each number from $\{1, \ldots n\}$ appears exactly once in each row and column---it is like a solved Sudoku without the condition on the smaller squares. 
A Latin square is also a particular kind of (unnormalised) magic square, as in the latter, the numbers need not be taken from a given set.  
Formally, $ L \in \mat_n(\setn)$ is a Latin square if each number from $\setn$ appears exactly once in each row and each column of $L$.

\emph{Quantum} Latin squares are defined  in \cite{Mu16b}; let us now review their definition. 

\begin{definition}[Quantum Latin squares] $L \in \mat_n(\C^n)$  
is a \emph{quantum Latin square} if every row and column of $L$ forms an orthonormal basis of $\C^n$.
\end{definition}

\begin{example}[Easy quantum Latin square] \label{exa:QLS}
The easiest way to construct a quantum Latin square is to take a Latin square and an orthonormal basis of the correct size, and to arrange the basis according to the indices in the Latin square. 
For example, given the $4 \times 4$ Latin square 
$$\begin{array}{|c|c|c|c|}
  \hline
  1 & 2 & 3 & 4 \\
  \hline
  2 & 4& 1&3 \\
  \hline
  3 & 1 & 4 &2 \\
  \hline
  4&3&2&1\\
  \hline
\end{array}
$$
and an orthonormal basis $v_1, \ldots,v_4 \in \C^4,$ we obtain the following quantum Latin square:
$$\begin{array}{|c|c|c|c|}
\hline
v_1 & v_2 & v_3 & v_4 \\
\hline 
v_2 & v_4 & v_1 & v_3 \\
\hline
v_3 & v_1 & v_4 & v_2\\
\hline
v_4&v_3&v_2&v_1 \\
\hline
\end{array}$$
\end{example}

There are more quantum Latin squares than the easy ones, as this example from \cite{Mu16b} shows:
$$\begin{array}{|c|c|c|c|}
\hline
v_1 & v_2 & v_3 & v_4 \\
\hline 
\tfrac{1}{\sqrt{2}}(v_2 - v_3)  & \tfrac{1}{\sqrt{5}}(iv_1 + 2v_4) & \tfrac{1}{\sqrt{5}}(2v_1 + iv_4) & \tfrac{1}{\sqrt{2}}(v_2 + v_3) \\
\hline
\tfrac{1}{\sqrt{2}}(v_2 + v_3) & \tfrac{1}{\sqrt{5}}(2v_1 + iv_4) &\tfrac{1}{\sqrt{5}}(iv_1 + 2v_4)& \tfrac{1}{\sqrt{2}}(v_2 - v_3) \\
\hline
v_4&v_3&v_2&v_1 \\
\hline
\end{array}$$
Here $v_1, \ldots,v_4 \in \C^n$ is again a fixed orthonormal basis. In this square,   four different orthonormal bases can be found in the rows and columns.

\section{From quantum Latin squares to quantum magic squares} 
\label{sec:QLS vs QMS}

In this section we investigate the relation between quantum Latin squares and quantum magic squares. 
We first explain how quantum Latin squares can be understood as quantum magic squares (\cref{ssec:qls}), 
show that the easy quantum Latin squares are precisely the semiclassical ones (\cref{ssec:easyquantumlatin}),  
study quantum Latin squares and friends (\cref{ssec:qlsfriends}), 
as well as these sets under the matrix convex hull  (\cref{ssec:tmconv}). 

\subsection{Quantum Latin squares as quantum magic squares}\label{ssec:qls}
 
We start by noting that quantum Latin squares are essentially equivalent to rank one quantum magic squares. 

\begin{observation} \label{prop:EQLS}
Let $V = (v_{ij})_{i,j=1}^n \in \mat_n(\C^n)$ be a quantum Latin square. Then 
$$\left(v_{ij} v_{ij}^* \right)_{i,j=1}^n \in \mat_n(\hern)$$ 
is a quantum magic square.
Conversely, given a quantum magic square $A = (A_{ij})_{i,j=1}^n \in \mat_n(\hern)$ with $\operatorname{rank}(A_{ij})=1$ for all $i,j$, there exist $a_{ij} \in \C^n$ such that  $A_{ij} = a_{ij}a_{ij}^*$  and $$(a_{ij})_{i,j=1}^n \in \mat_n(\C^n)$$ is a quantum Latin square.  
\end{observation}

Let us formalize this now, together with another way to construct quantum magic squares.

\begin{definition}[Rank one quantum magic squares]\label{def_qula}
($i$) 
The set of quantum magic squares of exterior size $n$ and interior size $s$, where each entry matrix has rank one, is denoted by 
$$
\mathsf{R}_s \coloneqq \{A \in \mathsf{M}_s^{(n)} \mid \mathrm{rank}(A_{ij})=1 \mbox{ for all  } i,j  \}. 
$$
and  $\mathsf{R} \coloneqq \bigcup_{s\in \N} \mathsf{R}_s$.

($ii$) 
The set of quantum Latin squares that arise from an easy quantum Latin square as in \cref{exa:QLS}  
and applying \Cref{prop:EQLS} 
is denoted by
\begin{align*}
\mathsf{L}_n:= \left\{\left(v_{L_{ij}}v_{L_{ij}}^*\right)_{i,j=1}^n  \mid  \right. & L \text{ classical Latin square of size } n, \\ 
&\ \left. v_1,\ldots, v_n \in \C^n \text{ orthonormal basis} \right\}.
\end{align*}
\end{definition}

The set $\textsf{R}_n^{(n)}$ is in one-to-one correspondence with the set of quantum Latin squares of size $n$ (up to a phase in the orthonormal bases, of course), as shown in  \Cref{prop:EQLS}. 
For $s>n$ the set $\textsf{R}_s^{(n)}$ is empty, since less than $s$ matrices of rank $1$ cannot sum up to the identity matrix $I_s$.

Every element from $\textsf{R}_n^{(n)}$ is a quantum permutation matrix, since for a unit vector $v \in \C^n,$ $vv^*$ is an orthogonal projection.
For $\textsf{R}_s^{(n)}$  with $s<n$ this is not the case, as the following example shows:
\begin{align*}
\begin{array}{|c|c|c|}
  \hline
  \tfrac{1}{2}e_1e_1^* & \tfrac{1}{2}e_1 e_1^* & e_2e_2^* \\
  \hline
   e_2e_2^* & \tfrac{1}{2}e_1e_1^*& \tfrac{1}{2}e_1e_1^* \\
  \hline
  \tfrac{1}{2}e_1e_1^* &  e_2e_2^* & \tfrac{1}{2}e_1e_1^* \\
  \hline
\end{array} \in \textsf{R}^{(3)}_2
\end{align*}
Here $e_1,e_2 \in \C^2$ is an orthonormal basis. 

\begin{remark} \label{rmk:EQLS-properties}
We have seen that 
$$
\textsf{L}_n \subseteq 
\textsf{R}_n^{(n)} \subseteq 
\textsf{P}_{n}^{(n)}\subseteq 
\textsf{M}_n^{(n)},  $$
and 
$$ 
\textsf{L}_n \subseteq 
\textsf{C}_n^{(n)} 
$$ 
because for two orthonormal vectors $v,w$, their associated rank 1 projectors commute.
\end{remark}

\subsection{Easy quantum Latin square are semiclassical} \label{ssec:easyquantumlatin}

Our first main result  shows that $\textsf{L}_n$ is precisely the set of quantum Latin squares that are semiclassical: 

\begin{theorem} \label{prop:sc and ro = Ln}
$\textsf{L}_n= \textsf{R}_n^{(n)} \cap \textsf{S}_n^{(n)}.$
\end{theorem}
\begin{proof}
First note that any element from $\textsf{L}_n$ is semiclassical, 
since  $\textsf{L}_n \subseteq \textsf{C}_n^{(n)} \subseteq \textsf{S}_n^{(n)}$. 
Since $\textsf{L}_n \subseteq  \textsf{R}_n^{(n)} $ too, this proves that the left hand side is included in the right hand side. 

For the other inclusion, let  $A = \sum_{\pi \in S_n} P_{\pi} \otimes Q_{\pi}$ be a semiclassical  magic square of interior and exterior size $n$ such that each entry has rank 1.
Hence every $Q_{\pi}$ has rank at most one, 
and thus the sum of two such $Q_\pi$  has rank at most 
one if and only if they are linearly dependent. 
Now consider 
$$
\Pi_{ij} \coloneqq \{\pi \in S_n \mid (P_{\pi})_{i,j} = 1 \}=\{ \pi \in S_n\mid \pi(j)=i\}.
$$
For $\pi, \tilde{\pi} \in \Pi_{ij}$ the above arguments imply that $Q_\pi$ and $Q_{\tilde\pi}$ are linearly dependent, so 
 \begin{equation}
 \label{eq_eins}\dim({\rm span}\{Q_{\pi} \mid \pi \in \Pi_{ij} \}) = 1 \quad \mbox{for all }i,j
 \end{equation} 
 Note that $A_{ij} \neq 0$ for all $i,j$, because each $A_{ij}$ is rank one and otherwise the corresponding row or column of $A$ could not sum up to $I_n$.

Now for every $i \in \{1,\ldots,n\}$ choose a  $\pi_i \in \Pi_{i1}$ with  $Q_{\pi_i} \neq 0$. 
By a similar rank argument as above, we see that $A_{i1} \notin {\rm span}\{A_{j1}\},$ and hence $Q_{\pi_i} \notin {\rm span}\{Q_{\pi_j}\}$ for $i\neq j$.
In view of (\ref{eq_eins}) this implies that each  $\Pi_{k\ell}$ contains at most one of the $\pi_i$.
So we have found permutations $\pi_1,\ldots,\pi_n\in S_n$ that are completely disjoint, in the sense that no two of them coincide on some input $\ell$. But then for each $k,\ell$ there must be some $i$ with $\pi_i(\ell)=k$. Hence  each $\Pi_{k\ell}$ contains precisely one of the $\pi_i$.

We now claim that $Q_{\pi} = 0$ for  $\pi \in S_n \setminus \{\pi_1,...,\pi_n\}$. 
 For $\ell=\pi(1)$ we have $\pi \in \Pi_{\ell 1}$, and since $\pi \neq \pi_{\ell}$ there exists some $j\neq 1$ with $\pi(j) \neq \pi_{\ell}(j)$. 
So $\pi \in \Pi_{\pi(j)j}$ and $\pi_\ell\notin \Pi_{\pi(j)j}$, and by the above considerations  we find an $i\neq \ell$ with $\pi_i \in \Pi_{\pi(j)j}$. This implies 
$$Q_{\pi} \in {\rm span}\{Q_{\pi_\ell}\} \cap {\rm span}\{Q_{\pi_i}\} = \{0\},
$$ 
as claimed.

Altogether we have shown $A = \sum_{i=1}^n P_{\pi_i} \otimes Q_{\pi_i}$, and  the $Q_{\pi_i}$ are all rank 1 squares, i.e.\ $Q_{\pi_i} = q_iq_i^*$ for certain $q_i\in\C^n$. 
 \Cref{prop:EQLS} implies that $q_1,\ldots, q_n$  form an orthonormal basis, and that 
 $A$ is the quantum Latin square constructed from this orthonormal  basis and the classical Latin square 
$$\sum_{j=1}^n  j \, P_{\pi_j}.$$ 
\end{proof}

\subsection{Quantum Latin squares and friends} \label{ssec:qlsfriends}

Instead of taking an orthonormal basis to construct an easy quantum Latin square, 
we can take  a POVM or a PVM, and  arrange its elements according to a classical Latin square of size $n$, similarly to  \Cref{exa:QLS}. 
These result in a certain type of quantum magic squares, 
where there is a \emph{single} POVM (or PVM) that is placed in a permuted way in every row and column. 
We denote them  $\mathsf{POVML}_s^{(n)}$ and $\mathsf{PVML}_s^{(n)}$, respectively. 

\begin{definition}[Quantum Latin squares and friends]\label{def:Latin}
We define
\begin{align*}
\mathsf{POVML}_s^{(n)}&:=\left\{\left(P_{L_{ij}}\right)_{i,j=1}^n \mid P_1,\ldots, P_n \in \hers \text{ POVM, } L \text{ Latin square} \right\} \\
\mathsf{PVML}_s^{(n)} &:= \left\{\left(P_{L_{ij}}\right)_{i,j=1}^n \mid P_1,\ldots, P_n \in \hers \text{ PVM, } L \text{ Latin square} \right\} 
\end{align*}
as well as 
$$
\mathsf{POVML}^{(n)} \coloneqq \bigcup_{s\in \N} \mathsf{POVML}_s^{(n)}, \quad 
\mathsf{PVML}^{(n)} \coloneqq \bigcup_{s \in \N} \mathsf{PVML}_s^{(n)}.
$$
\end{definition}

\begin{proposition}\label{l:pov is sc}
$\mathsf{POVML}^{(n)} \subseteq \mathsf{S}^{(n)}$ and 
$\mathsf{PVML}^{(n)}\subseteq \mathsf{C}^{(n)}\cap\mathsf{POVML}^{(n)}.$
\end{proposition}

\begin{proof}
Let $A\in\mathsf{POVML}_s^{(n)}$ be generated by the POVM $Q_1,\ldots, Q_n \in \hers$ and the  Latin square $L$. 
We define permutation matrices $P_1,\ldots, P_n$ by setting for all 
$i,k,\ell \in \{1,\ldots n\}$:
\begin{align*}
(P_i)_{k,\ell} = \begin{cases} 1 &\mbox{if } L_{k\ell} = i \\
0 & \mbox{else.} \end{cases}
\end{align*}
Then clearly 
$$ A = \sum_{i=1}^n P_i \otimes Q_i,$$
which shows that $A$ is semiclassical and proves the first inclusion.
The second inclusion is obvious, since projectors from one PVM always commute.
\end{proof}

The inclusions   observed so far are shown in \Cref{fig:inc}.

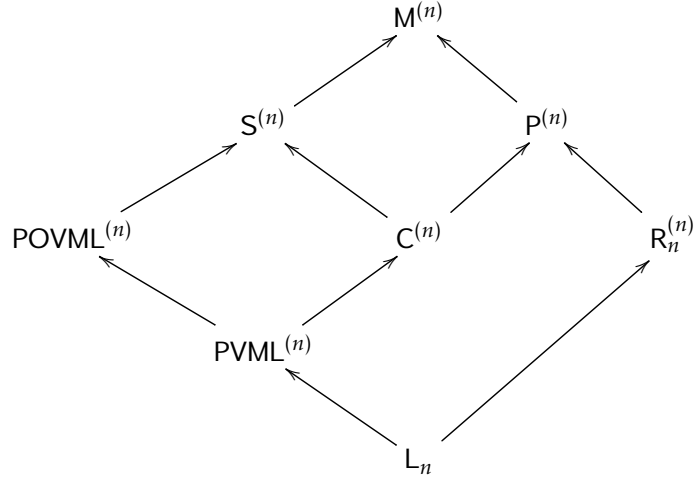
\begin{figure}[h!]
\begin{center}
\hfil
\xymatrix{
&&\mathsf{M}^{(n)}&&\\
&\mathsf{S}^{(n)}\ar@{->}[ur]&&\mathsf{P}^{(n)}\ar@{->}[ul]&\\
\mathsf{POVML}^{(n)}\ar@{->}[ur]&&\mathsf{C}^{(n)}\ar@{->}[ul]\ar@{->}[ur]&&\mathsf{R}_{n}^{(n)}\ar@{->}[ul]\\
&\mathsf{PVML}^{(n)}\ar@{->}[ul]\ar@{->}[ur]&&&\\
&&\mathsf{L}_n\ar@{->}[ul]\ar@{->}[uurr]&&
}
\hfil
\end{center}
\caption{Inclusions between sets of quantum magic squares, where  $\mathsf{A} \to \mathsf{B}$ denotes $\mathsf{A}\subseteq \mathsf{B}$. 
That easy quantum Latin squares are PVM Latin squares, which are POVM Latin squares, 
$\mathsf{L}_n \subseteq \mathsf{PVML}^{(n)} \subseteq \mathsf{POVML}^{(n)}$, is the case by definition,
and $\mathsf{POVML}^{(n)} \subseteq \mathsf{C}^{(n)}$ by \cref{l:pov is sc}. 
That commuting quantum permutation matrices are quantum permutation matrices, which are quantum magic squares, 
$\mathsf{C}^{(n)} \subseteq \mathsf{P}^{(n)} \subseteq \mathsf{M}^{(n)}$, 
is true by definition.
That quantum Latin squares are rank one, which are quantum permutation matrices, 
$ \textsf{L}_n \subseteq \textsf{R}_n^{(n)} \subseteq \textsf{P}_{n}^{(n)}$,
is stated in \cref{rmk:EQLS-properties}. 
Finally, 
$ \textsf{C}^{(n)}  \subseteq \textsf{S} ^{(n)} \subseteq \textsf{M}^{(n)}$ was shown in \cite{J},
and  $\mathsf{POVML}^{(n)} \subseteq \textsf{S}^{(n)} $ by \cref{l:pov is sc}. 
}
\label{fig:inc}
\end{figure}

\subsection{Taking the matrix convex hull}\label{ssec:tmconv}

Let us now examine the matrix convex hulls of some quantum magic squares.

\begin{theorem}\label{thm_mconv}
\begin{align*}\mathsf{S}^{(n)} &=  
{\mconv}(\mathsf{C}^{(n)})= 
{\mconv}(\mathsf{POVML}^{(n)}) = 
{\mconv}(\mathsf{PVML}^{(n)}) = 
{\mconv}(\mathsf{L}_n).
\end{align*}
\end{theorem}

\begin{proof}
In view of the inclusions of \cref{fig:inc} and Eq.\  \eqref{eq:mconvPn1} it suffices to prove that 
$\mathsf{P}_1^{(n)}\subseteq \mconv(\mathsf{L}_n)$. 
So let $P\in\mathsf{P}_1^{(n)}$ be a  permutation matrix, and denote by $L$ the classical Latin square obtained from $P$ by replacing its zero entries with numbers $2,\ldots n$ suitably. 
Let $v_1,\ldots, v_n$ be an orthonormal basis of $\C^n$, and let $A\in \mathsf{L}_n$ be the easy quantum Latin square constructed from $L$ and this basis. 
Then $P=v_1^*Av_1\in \mconv(\mathsf{L}_n)$, which finishes the proof.
\end{proof}

This shows that the left hand side of  \cref{fig:inc}  simplifies greatly after taking the matrix convex hull. Let us now consider what happens to the right hand side of   \cref{fig:inc} under the matrix convex hull. 

First note that from $\mathsf{L}_n \subseteq \mathsf{R}_n^{(n)}$ we immediately obtain 
$$
\mconv(\mathsf{L}_n) 
\subseteq \mconv(\mathsf{R}_n^{(n)}).
$$ 

\begin{theorem}\label{thm_ntqls}
For even $n\geqslant 4$ we have $\mconv(\mathsf{L}_n) \subsetneq \mconv(\mathsf{R}_n^{(n)})$, and in particular $\mathsf{L}_n\subsetneq \mathsf{R}_n^{(n)},$ i.e.\ there exist quantum Latin squares which are not semiclassical. 
\end{theorem}

\begin{proof}
Write $n=2m$, let $L$ be a classical Latin square of size $m,$ and let $v_1,\ldots,v_m \in \C^m, w_1,\ldots,w_m \in \C^{m}$ be two orthonormal bases  such that 
$$
v_1v_1^*w_1w_1^* \neq w_1w_1^*v_1v_1^*.
$$ 
Let $A,B \in \mathsf{L}_{m}$ be the quantum magic squares generated by $L$ and these two bases, respectively. The following block matrix is a quantum permutation matrix of size $n$:
$$ 
A {\oplus} B \coloneqq
\begin{pmatrix}
A & 0 \\ 0 & B
\end{pmatrix} \in \mathsf{P}^{(n)} . 
$$
For example, for $n=4 $ and $L=\left(\begin{array}{cc}1 & 2 \\2 & 1\end{array}\right)$, 
this is the matrix
$$
\begin{pmatrix}
v_1v_1^* &v_2v_2^* &0&0 \\
v_2v_2^* &v_1v_1^* & 0&0 \\
0&0&w_1w_1^* & w_2w_2^* \\
0&0&w_2 w_2^* & w_1w_1^*
\end{pmatrix}.
$$
By our choice of bases we have that $A{\oplus} B \notin \mathsf{C}^{(n)}$, 
and the proof of \cite[Corollary 15]{J}  shows that $A {\oplus} B \notin \mathsf{S}^{(n)}$. 

Now we claim that $A {\oplus} B \in \mconv(\mathsf{R}_n^{(n)})$. To see this, let 
$$\iota: \C^{m} \to \C^n : v \mapsto \begin{pmatrix}
v \\ 0
\end{pmatrix}.$$
Then $\iota(v_1),\ldots,\iota(v_m)$ and $\iota(w_1),\ldots,\iota(w_m)$ are orthonormal systems in $\C^n,$ which are both extended to an orthonormal basis by the standard basis vectors $e_{m+1},\ldots, e_n$.
Now let $\iota(A)$ be the matrix constructed from  $\iota(v_1),\ldots,\iota(v_m)$ and $L$ as in \Cref{def_qula} ($i$) (except that  $\iota(v_1),\ldots,\iota(v_m)$ only form half of an orthonormal basis here). In the same way we construct  $\iota(B)$ from $\iota(w_1),\ldots,\iota(w_m)$  and $L$, as well as $E$ from $e_{m+1},\ldots e_n$ and $L$.
We then obtain
$$C = \begin{pmatrix}
\iota(A) & E \\ E & \iota(B)
\end{pmatrix} \in \mathsf{R}_n^{(n)},
$$
since  each row and each column contains the rank one squares of an orthonormal basis of $\C^n.$ 
For $n=4$ this looks like 
$$
\begin{pmatrix}
\iota(v_1)\iota(v_1)^* & \iota(v_2)\iota(v_2)^* &e_3 e_3^*&e_4 e_4^* \\
\iota(v_2) \iota(v_2)^* & \iota(v_1) \iota(v_1)^* & e_4 e_4^*&e_3 e_3^* \\
e_3 e_3^*&e_4 e_4^*& \iota(w_1)\iota(w_1)^* & \iota(w_2)\iota(w_2)^* \\
 e_4 e_4^*&e_3 e_3^*& \iota(w_2)\iota(w_2)^* & \iota(w_1) \iota(w_1)^*
\end{pmatrix}.
$$
For $V = \begin{pmatrix} I_m \\ 0 \end{pmatrix} \in \mat_{n,m}(\C)$ we have $V^*V = I_{m}$ and 
$$
V^*CV= A \oplus B \in \mconv(\mathsf{R}_n^{(n)}),
$$ 
which proves the claim.
\end{proof}

It follows that, after taking matrix convex hulls, the diagram from \cref{fig:inc} simplifies to \cref{fig_mconv}.

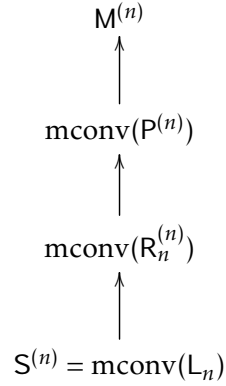
\begin{figure}[t]
\begin{center}
\hfil
\xymatrix{ 
\mathsf{M}^{(n)}\\
 \mconv(\mathsf{P}^{(n)})\ar@{->}[u]\\
\mconv(\mathsf{R}_n^{(n)})\ar@{->}[u]\\
\mathsf{S}^{(n)}  = \mconv(\mathsf{L}_n)\ar@{->}[u]
}
\hfil
\end{center}
\caption{Inclusions between the matrix convex hulls (compare with \cref{fig:inc}). 
$\mathsf{S}^{(n)}  = \mconv(\mathsf{L}_n)$ is shown in \cref{thm_mconv},
and its inclusion in $\mconv(\mathsf{R}_n^{(n)})$ is strict for even $n\geqslant 4$ by \cref{thm_ntqls}. 
The inclusion of $\mconv(\mathsf{P}^{(n)})$ in $\mathsf{M}^{(n)}$ is strict for all $n\geqslant 3$, as proven in \cite{J}. We do not know if  the inclusion $\mconv(\mathsf{R}_n^{(n)}) \subseteq \mconv(\mathsf{P}^{(n)})$ is strict. }
\label{fig_mconv}
\end{figure}

\begin{remark}
We thank David E.\ Roberson for bringing to our attention the following fact.
The statement $\mathsf{L}_n\subsetneq \mathsf{R}_n^{(n)}$ holds for \emph{all} $n\geq 4$, 
not only the even ones as proven in \cref{thm_ntqls}. 
This follows from \cite[Theorem 4.6]{Ro20b}, which shows that for each $n$ there are $n!$ quantum Latin squares of which at most $n(n-1)$ have commuting entries.
This implies that for every $n\geq 4$ there are quantum Latin squares not in $\mathsf{C}^{(n)}$ and thus neither in $\mathsf{L}_n$. 
\end{remark}

\section{Conclusions and Outlook}\label{sec:conclusions}

In this work, 
we have proven three results concerning the structure of the special types of quantum magic squares. 
First, every semiclassical magic square can be purified to a quantum Latin square;  
equivalently, 
the matrix convex hull of several sets of Latin squares is
the set of  semiclassical magic squares (\Cref{thm_mconv}). 
Second, the matrix convex hull of all quantum Latin squares is  larger than the set of semiclassical magic squares;  this follows from the existence of non-semiclassical quantum Latin squares of even sizes larger than $2$ (\Cref{thm_ntqls}). 
And third, 
the easy quantum Latin squares (i.e.\ those that arise from a classical Latin square) 
are precisely the semiclassical ones (\Cref{prop:sc and ro = Ln}).

It would be interesting to define and characterise quantum magic cubes, as well as the three `dimensional' version of their friends. 
Quantum permutation cubes would be represented by the hypergraph of \cref{fig:magiccube}. 
Not every magic cube can be expressed as a convex combination of permutation cubes, 
i.e.\ the analogue of Birkhoff--von Neumann theorem fails. 
Perhaps a quantum version of a magic cube could recover a behaviour of this type.

\begin{figure}
\includegraphics[width=0.2\columnwidth]{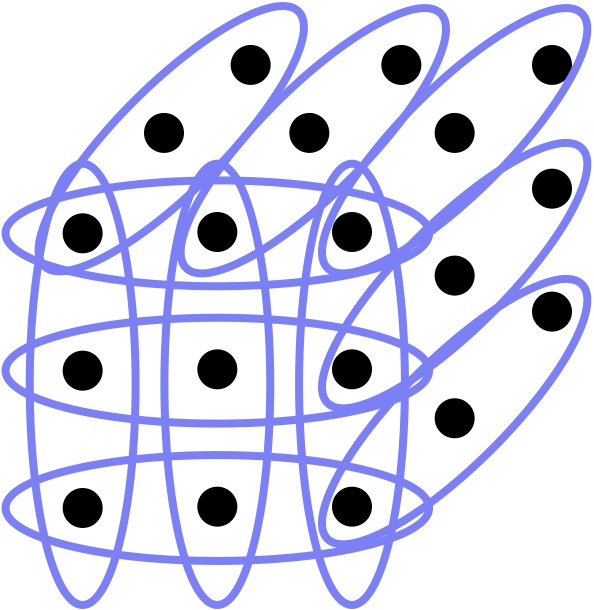}
\caption{The quantum hypergraph corresponding to a quantum magic cube. Nodes and hypergraphs behind the top layer are now shown. }
\label{fig:magiccube}
\end{figure}

Quantum magic squares can be seen as `magical' combinations of quantum measurements, and this work and others investigate when this combination can be purified, as well as other properties. 
Positive semidefinite matrices are intimately related to completely positive maps by the Choi--Jamio\l{}kowski isomorphism (given a bipartite structure of the former). 
Yet, we do not know if any of the investigations on quantum magic squares has any consequences for `magical' combinations of completely positive maps. 

Finally, it would be very interesting to establish a correspondence between (quantum) magic squares and the study of contextuality \cite{Sc20}, 
and to understand matrix convex hulls of classical magic squares as epistemic restrictions of classical theories.


\begin{thebibliography}{10}

\bibitem{bailey_2008}
{\sc Bailey, R.~A.}
\newblock {\em Design of Comparative Experiments}.
\newblock Cambridge Series in Statistical and Probabilistic Mathematics.
  Cambridge University Press, 2008.

\bibitem{Bl21}
{\sc Blakaj, V., and Wolf, M.~M.}
\newblock {Transcendental properties of entropy-constrained sets}.
\newblock {\em Ann. Henri Poincar{\'{e}}\/} (2022).

\bibitem{Bl18}
{\sc Bluhm, A., and Nechita, I.}
\newblock {Joint measurability of quantum effects and the matrix diamond}.
\newblock {\em J. Math. Phys. 59\/} (2018), 112202.

\bibitem{Bu21}
{\sc Buckley, A.}
\newblock {New examples of entangled states on $\mathbb{C}^3\otimes
  \mathbb{C}^3$}.
\newblock {\em arXiv:2112.12643\/}.

\bibitem{Cu21}
{\sc Cuffaro, M.~E., and Hartmann, S.}
\newblock {The open systems view}.
\newblock {\em arXiv:2112.11095\/} (2021).

\bibitem{J}
{\sc De~las Cuevas, G., Drescher, T., and Netzer, T.}
\newblock Quantum magic squares: Dilations and their limitations.
\newblock {\em J. Math. Phys. 61}, 11 (2020), 111704.

\bibitem{De21d}
{\sc {De las Cuevas}, G., and Netzer, T.}
\newblock {Quantum Information Theory and Free Semialgebraic Geometry: One
  Wonderland Through Two Looking Glasses}.
\newblock {\em IMN\/} (2021), 246.

\bibitem{Fr16b}
{\sc Fritz, T.}
\newblock {Quantum logic is undecidable}.
\newblock {\em Arch. Math. Logic 60\/} (2021), 329.

\bibitem{Mconvfreesets}
{\sc Helton, J.~W., Klep, I., and McCullough, S.}
\newblock Matrix convex hulls of free semialgebraic sets.
\newblock {\em Trans. Amer. Math. Soc. 368}, 5 (2015), 3105.

\bibitem{Ji20}
{\sc Jivulescu, M.~A., Lancien, C., and Nechita, I.}
\newblock {Multipartite entanglement detection via projective tensor norms}.
\newblock {\em Ann. Henri Poincar{\'{e}}\/} (2022).

\bibitem{UEB}
{\sc Klappenecker, A., and R{\"o}tteler, M.}
\newblock Unitary error bases: Constructions, equivalence, and applications.
\newblock In {\em International Symposium on Applied Algebra, Algebraic
  Algorithms, and Error-Correcting Codes\/} (2003), Springer, p.~139.

\bibitem{Lu17b}
{\sc Lupini, M., Man{\v{c}}inska, L., and Roberson, D.~E.}
\newblock {Nonlocal Games and Quantum Permutation Groups}.
\newblock {\em J. Funct. Anal. 279\/} (2020), 108592.

\bibitem{Me09b}
{\sc Mendl, C.~B., and Wolf, M.~M.}
\newblock {Unital Quantum Channels -- Convex Structure and Revivals of
  Birkhoff's Theorem}.
\newblock {\em Comm. Math. Phys. 289\/} (2009), 1057.

\bibitem{Mu16b}
{\sc Musto, B., and Vicary, J.}
\newblock {Quantum Latin squares and unitary error bases}.
\newblock {\em Quantum Inf. Comput. 16\/} (2016), 1318.

\bibitem{Ne20}
{\sc Nechita, I., and Pillet, J.}
\newblock {SudoQ -- a quantum variant of the popular game}.
\newblock {\em arXiv:2005.10862\/} (2020).

\bibitem{paulsen2002completely}
{\sc Paulsen, V.}
\newblock {\em Completely bounded maps and operator algebras}.
\newblock Cambridge University Press, 2002.

\bibitem{Ra21}
{\sc Rather, S.~A., Burchardt, A., Bruzda, W., Rajchel-Mieldzio{\'{c}}, G.,
  Lakshminarayan, A., and K.{\.{Z}}yczkowski}.
\newblock {Thirty-six entangled officers of Euler: Quantum solution to a
  classically impossible problem}.
\newblock {\em Phys. Rev. Lett. 128\/} (2022), 080507.

\bibitem{Ro20b}
{\sc Roberson, D.~E., and Schmidt, S.}
\newblock {Quantum symmetry vs nonlocal symmetry}.
\newblock {\em arxiv:2012.13328\/} (2020).

\bibitem{Sc20}
{\sc Schmid, D., Selby, J.~H., and Spekkens, R.~W.}
\newblock {Unscrambling the omelette of causation and inference: The framework
  of causal-inferential theories}.
\newblock {\em arXiv:2009.03297\/} (2020).

\bibitem{Sm06}
{\sc Smith, J.}
\newblock {\em {An Introduction to Quasigroups and Their Representations}}.
\newblock Chapman and Hall, 2006.

\end{thebibliography}

\end{document}